\documentclass[final]{dmtcs-episciences}

\usepackage[utf8]{inputenc}
\usepackage{xspace}
\usepackage{amsmath}
\usepackage{amsthm}
\newcommand {\mathset} [1] {\ensuremath {\mathbb {#1}}\xspace}
\newcommand {\Q} {\mathset {Q}}
\renewcommand{\S}{\ensuremath{\mathcal{S}}\xspace}
\newcommand{\ACMQ}{ACM$\Q$\xspace}

\newcommand{\SPM}{\ensuremath{\mathit{SPM}}\xspace}
\newtheorem{conjecture}{Conjecture}
\newtheorem{corollary}{Corollary}
\newtheorem{lemma}{Lemma}
\newtheorem{observation}{Observation}
\newtheorem{theorem}{Theorem}
\newcommand{\forceqed}{\tag*{\qedsymbol}}
\allowdisplaybreaks
\makeatletter
\renewcommand{\@makecaption}[2]{%
  \vskip\abovecaptionskip
  {\fontsize{9.7}{10}\selectfont \textbf{#1.} #2\par}
  \vskip\belowcaptionskip
}
\makeatother

\usepackage[round]{natbib}

\author[Sarita de Berg et al.]{Sarita de Berg\affiliationmark{1,2}
  \and Guillermo Esteban\affiliationmark{3}\\
  \and Rodrigo I. Silveira\affiliationmark{4} \and Frank Staals\affiliationmark{2}}
\title{Exact solutions to the Weighted Region Problem}
\affiliation{
  IT University of Copenhagen, Copenhagen, Denmark\\
  Utrecht University, Utrecht, The Netherlands\\
  Universidad de Alcalá, Alcalá de Henares, Spain\\
  Universitat Politècnica de Catalunya, Barcelona, Spain}
\begin{document}
\publicationdata{vol. 28:2}{2026}{25}{10.46298/dmtcs.15053}{2025-01-10; 2025-01-10; 2026-01-12}{2026-03-31}
\maketitle
\vspace*{-0.15cm}
\begin{abstract}
  In this paper, we consider the Weighted Region Problem. In the Weighted Region Problem, the length of a path is defined as the sum of the weights of the subpaths within each region, where the weight of a subpath is its Euclidean length multiplied by a weight $ \alpha \geq 0 $ depending on the region. We study a restricted version of the problem of determining shortest paths through a single weighted rectangular region. We prove that even this very restricted version of the problem is unsolvable within the Algebraic Computation Model over the Rational Numbers (ACM$ \mathbb{Q} $). On the positive side, we provide equations for the shortest paths that are computable within the ACM$ \mathbb{Q} $. Additionally, we provide equations for the bisectors between regions of the Shortest Path Map for a source point on the boundary of (or inside) the rectangular region.
\end{abstract}
\section{Introduction}
The Weighted Region Problem (WRP)~\cite{mitchell1991weighted} is a well-known geometric problem that, despite having been studied extensively, is still far from being well understood. Consider a subdivision of the plane into (usually polygonal) regions. Each region $ R_i $ has a  weight $ \alpha_i \geq 0$, representing the cost per unit distance of traveling in that region. Thus, a straight-line segment~$ \sigma $, of Euclidean length $|\sigma|$, between two points in the same region has weighted length $ \alpha_i \cdot \rvert \sigma \lvert $ when traversing the interior of $ R_i $, or $ \min\{\alpha_i,\alpha_j\}\cdot\rvert \sigma \lvert $ if it goes along the edge between~$ R_i $ and $ R_j $. Then, the weighted length of a path through a subdivision is the sum of the weighted lengths of its subpaths through each face or edge. The resulting metric is called the \emph{Weighted Region Metric}.
The WRP entails computing a shortest path $ \pi(s,t) $ between two given points $s$ and~$t$ under this metric. We denote the weighted length of $ \pi(s,t) $ by~$d(s,t)$. 
Figure~\ref{fig:example} shows how the shape of a shortest path changes as the weight of one region varies.

\begin{figure}[tb]
    \centering
    \includegraphics[page=6]{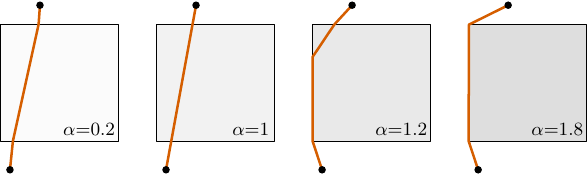}
    \caption{Examples of shortest paths between two points---shown in orange---for two weighted regions. The unbounded region has weight 1, the squares have varying weight $\alpha$.}
    \label{fig:example}
\end{figure}

Existing algorithms for the WRP---in its general formulation---are approximate. Since the seminal work by~\cite{mitchell1991weighted}, with the first $(1+\varepsilon)$-ap\-prox\-i\-ma\-tion, several  algorithms have been proposed, with improvements on running times, but always keeping some dependency on the vertex coordinates sizes and weight ranges. These methods are usually based on the continuous Dijkstra's algorithm, subdividing triangle edges in parts for which crossing shortest paths have the same combinatorial structure (e.g.,~\cite{mitchell1991weighted}), or on adding Steiner points (e.g., see~\cite{aleksandrov1998varepsilon,aleksandrov2000approximation,aleksandrov2005determining,Equilateral,sun2001bushwhack}). However, rather recently it has been proved that computing an exact shortest path between two points using the Weighted Region Metric, even if there are only three different weights, is an unsolvable problem in the Algebraic Computation Model over the Rational Numbers (\ACMQ)~\cite{carufel14}. In the \ACMQ one can  compute exactly any number that can be obtained from rational numbers by applying a finite number of operations from $ +, -, \times, \div $, and $ \sqrt[k]{} $, for any integer $ k \geq 2 $. This provides a theoretical explanation for the lack of exact algorithms for the WRP, and justifies the study of approximation methods. 

This also raises the question of which are the special cases for which the WRP can be solved exactly. Two natural ways to restrict the problem are by limiting the possible weights and by restricting the shape of the regions. For example, computing a shortest path among polygonal or curved obstacles can be seen as a variant of the WRP with weights in the set $\{1, \infty \}$. Efficient algorithms exists for this problem, culminating with the recent algorithms by~\cite{wang2023new} for polygonal obstacles, and by~\cite{hershberger2022near} for shortest paths among curved obstacles. The case for polygonal regions with weights in $ \{0, 1, \infty\} $ can be solved in $ O(n^2)$ time~\cite{gewali1988path} by constructing a graph known as the \emph{critical graph}, an extension of the visibility graph. Other variants that can be solved exactly correspond to regions shaped as regular $k$-gons with weight $\geq 2$ (since they can be considered as obstacles), or regions with two weights $ \{1, \alpha\} $ consisting of parallel strips~\cite{narayanappa2006exact}. In the latter case, the angle of incidence in each of the strips is the same, so they can be grouped into one strip that can be shifted so that $ s $ or $ t $ are on the boundary of the new strip, then the angle of incidence can be computed exactly using Snell's law of refraction.

\subparagraph{Our results.}
In light of the fact that the WRP is unsolvable in the \ACMQ already for three different weights, in this work we study the case of two arbitrary weights, that is, weights in $ \{1, \alpha\} $, where $ \alpha \in \mathbb{Q}^+ $. In particular, and without loss of generality, we assume that the weight of the unbounded region is $ 1 $. Otherwise, we could always rescale the weights to be $ 1 $ outside the regions. This case is particularly interesting, since an algorithm for weights  $ \{1, \alpha\} $ can be transformed into one for weights in $ \{0, 1, \alpha, \infty\} $~\cite{mitchell1987shortest}. However, the variant with weights $ \{1, \alpha\} $ was conjectured to be as hard as the general WRP problem, see the first open problem in \cite[Section $ 7 $]{gewali1988path}. (The results in~\cite{carufel14} do not directly apply to weights $ \{0, 1, \alpha, \infty\} $.). This paper is organized as follows. First we present some preliminaries in Section \ref{sec:Shortest_paths_and_their_properties}. In Section \ref{sec:Computing_a_shortest_path} we consider two weights and one rectangular region $R$, with the source point~$s$ on its boundary or inside. For this setting, we figure out all types of possible optimal paths and give exact formulas to compute their lengths. In Section~\ref{sub:source_point_outside} we focus on the case where $s$ is outside of $R$, and prove that in this case the WRP with weights $ \{1, \alpha\} $ is already unsolvable in the \ACMQ, confirming the suspicions of~\cite{mitchell1987shortest}. In Section~\ref{sec:Computing_a_shortest_path_map} we explore the computation of the Shortest Path Map for $s$. We finish with some conclusions in Section~\ref{sec:conclusions}.

\section{Shortest paths and their properties}
\label{sec:Shortest_paths_and_their_properties}

In this section we briefly review some key properties of shortest paths in weighted regions.

First, with our assumption that the weight within each region does not account for the effect of certain force fields that favors some directions of travel, there will always be a shortest path in the Weighted Region Problem that is piecewise linear, see~\cite[Lemma $ 3.1 $]{mitchell1991weighted}. Second, it is known that shortest paths must obey Snell’s law of refraction. So we can think of a shortest path as a ray of light. Throughout this paper, the \emph{angle of incidence} $ \theta $ is defined as the minimum angle between the incoming ray and the vector perpendicular to the region boundary. The \emph{angle of refraction} $ \theta' $ is defined as the minimum angle between the outgoing ray and the vector perpendicular to the region boundary. Snell's law states that whenever the ray goes from one region $ R_i $ to another region $ R_j $, then $\alpha_i \sin \theta = \alpha_{j} \sin \theta'$. In addition, whenever $\alpha_i > \alpha_j$, the angle $ \theta_c $ at which $ \frac{\alpha_i}{\alpha_{j}}\sin{\theta_c} = 1 $ is called the \emph{critical angle}. A ray that hits an edge at an angle of incidence greater than $ \theta_c $, will be totally reflected from the point at which it hits the boundary. In our problem, a shortest path will never be incident to an edge at an angle greater than~$ \theta_c $. 

Finally, if the space only contains orthoconvex regions\footnote{A region is orthoconvex if its intersection with every horizontal and vertical line is connected or empty~\cite{rawlins1988ortho}.} with weight at least~$ \sqrt{2} $, they can be simply considered as obstacles~\cite{narayanappa2006exact}. Thus, since we focus on a rectangular region $ R $, we assume that its weight is $ 0 < \alpha < \sqrt{2} $. We first provide some general properties of shortest paths for arbitrary weighted regions that are interesting on their own.

\begin{lemma}
    \label{prop:visitonce}
    Let $\S$ be a polygonal subdivision for which each region has a weight in the set $\{1, \alpha\}$, with $\alpha \geq 0$. A shortest path $\pi(s,t) $ visits any edge of the subdivision at most once.
\end{lemma}

\begin{proof}
    Assume, for the sake of contradiction, that $\pi(s,t)$ intersects an edge $e$ in at least two disjoint intervals $I_1$ and $I_3$ (note that $ I_1 $ and $ I_3 $ could be points). Moreover, let $p_1 \in I_1$ and $p_3 \in I_3$ be points for which the subpath $\pi(p_1,p_3) \subseteq \pi(s,t)$ does not intersect $e$ in any points other than $p_1$ and $p_3$. Let $p_2$ be a point on $\pi(p_1,p_3)$ between $p_1$ and $p_3$, which thus does not lie on $e$. Now observe that there exists a path $\overline{p_1p_3}$ from $p_1$ to $p_3$ of length $\min\{1,\alpha\}\lvert\overline{p_1p_3}\rvert$. Since $p_2$ does not lie on $\overline{p_1p_3}$, it follows by the triangle inequality that the length of $\pi(p_1,p_3)$ is strictly larger than $\min\{1,\alpha\}\lvert\overline{p_1p_3}\rvert$. Hence, $\pi(s,t)$ is not a shortest path, and we obtain a contradiction.
\end{proof}

Observe that the previous result is not true when there are more than two weights, see~\cite[Figure~$2$]{mitchell1991weighted}.

\begin{corollary}
    \label{cor:combinatorial_complexity_shortest_path}
    Let $\S$ be a polygonal subdivision with $n$ vertices, such that the weight of each region is within the set $\{1, \alpha\}$, with $\alpha \geq 0$. Any shortest path $\pi(s,t)$ is a polygonal chain with at most $O(n)$ vertices.
\end{corollary}

\begin{proof}
    Any shortest path is a polygonal chain whose interior vertices all lie on edges of~\S, see~\cite[Proposition $ 3.8 $]{mitchell1991weighted}. By Lemma~\ref{prop:visitonce}, each edge contributes at most two vertices.
\end{proof}

\section{Computing a shortest path}
\label{sec:Computing_a_shortest_path}

We now consider the problem of computing a shortest path $\pi(s,t)$ from $s$ to $t$ when the region $R$ is an axis-aligned rectangle of weight $\alpha$. The exact shape of $\pi(s,t)$ depends on the position of $s$ and $t$ with respect to~$ R $, and on the value of~$ \alpha $. In Sections~\ref{sub:The_source_point_lies_on_the_boundary} and \ref{sub:The_source_point_lies_inside} we consider the case where $s$ lies on the boundary or inside of~$R$, respectively. We categorize the various types of shortest paths, and show that we can compute the shortest path of each type, and thus we can compute $ \pi(s,t) $. In Section~\ref{sub:source_point_outside}, we consider the case where $s$ and $ t $ lie outside $R$. In this case $ \pi(s,t) $ may have only two vertices on the boundary of $R$, and these vertices may not have the critical angle property. We show that the coordinates of these vertices cannot be computed exactly within the \ACMQ.

\subsection{The source point $s$ lies on the boundary of $R$}
\label{sub:The_source_point_lies_on_the_boundary}

Throughout this section we consider the case where $s$ is restricted to the boundary of~$R$, a rectangle of unit height with top-left corner at $(0,0)$. Let $ s = (s_x,0)$, $s_x > 0$, be a point on the top side of~$R$, see Figure~\ref{fig:on_boundary}. In addition, we assume that $t$ is to the left of the line through $s$ perpendicular to the top side of~$R$. The other cases are symmetric.

\subparagraph{Shortest path types.}

Lemma~\ref{prop:visitonce} implies that in this setting, there are only $O(1)$ combinatorial types of paths that we have to consider. More precisely, we have that:

\begin{figure}[tb]
    \centering
    \includegraphics[page=3]{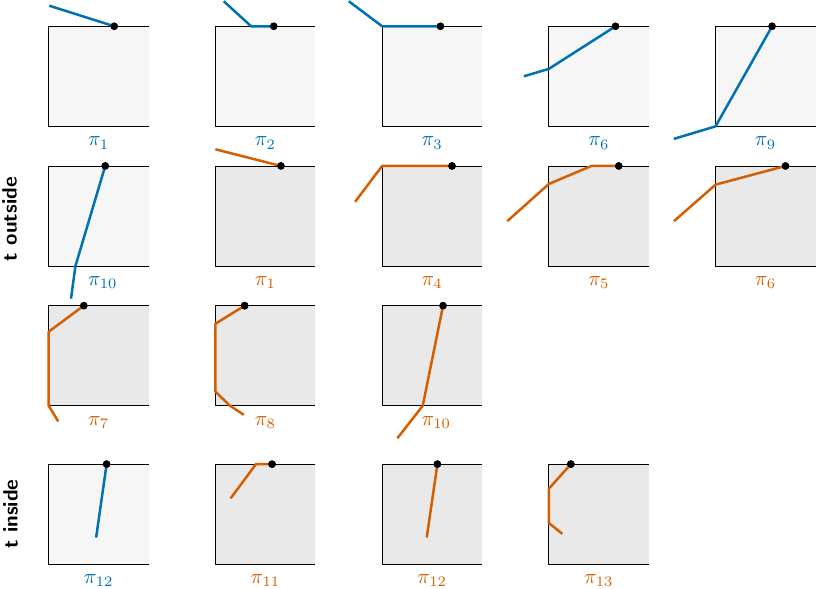}
    \caption{Path types for $s$ on the boundary of $R$ of weight $\alpha < 1$ (blue) and $1 \leq \alpha < \sqrt{2}$ (orange).}
    \label{fig:on_boundary}
\end{figure}

\begin{observation}
    \label{obs:typesboundary}
    Let $ s $ be a point on the top boundary of a rectangle $ R $ with weight~$ 0 <\alpha<\sqrt{2} $. There are $ 13 $ types of shortest paths $ \pi_i(s,t) $, shown in Figure~\ref{fig:on_boundary}, up to symmetries.
\end{observation}

Note that only some of the types can exist for both $\alpha < 1$ and $1 \leq \alpha < \sqrt{2}$. These types are included twice in Figure~\ref{fig:on_boundary}, once for each regime of $\alpha$.

\subparagraph{Length of $ \pi_i(s,t) $.}
When $s$ is on the boundary of $R$, there is at most one vertex of $ \pi_i(s,t) $ without the critical angle property. This allows us to compute the exact coordinates of the vertices of $ \pi_i(s,t) $ in the \ACMQ. We now provide the equations for the length $ d_i(s,t) $ of the $ 13 $ types of shortest paths $ \pi_i(s,t) $. Theorem~\ref{thm:length} summarizes the results.

\begin{theorem}
    \label{thm:length}
    Let $ s=(s_x,0) $ be a point on the boundary of $R$ with weight $ 0 < \alpha < \sqrt{2} $. A shortest path $\pi(s,t)=\pi_i(s,t) $ from $s$ to a point $ t=(t_x,t_y) $ and its length $d(s,t)=d_i(s,t)$ can be computed in $O(1)$ time in the \ACMQ. In particular, the length $d_i(s,t)$ is given by
    
    \begin{align*}
      d_i(s,t) = &
      \begin{cases}
        \sqrt{(s_x-t_x)^2+t_y^2}                                & \text{if } i = 1 \\
        \alpha(s_x-t_x)+\sqrt{1-\alpha^2}t_y                    & \text{if } i = 2 \\
        \alpha s_x+\sqrt{t_x^2+t_y^2}                           & \text{if } i = 3 \\
        s_x + \sqrt{t_x^2+t_y^2}                                & \text{if } i = 4 \\
        s_x - \sqrt{2-\alpha^2}t_x - \sqrt{\alpha^2-1}t_y       & \text{if } i = 5 \\
        \alpha\sqrt{s_x^2+w_1^2}+\sqrt{t_x^2+(t_y-w_1)^2}           & \text{if } i = 6 \\
        \sqrt{\alpha^2-1}s_x + 1 + \sqrt{t_x^2+(t_y+1)^2}       & \text{if } i = 7 \\
        \sqrt{\alpha^2-1}(s_x+t_x) - \sqrt{2-\alpha^2}(1+t_y)+1 & \text{if } i = 8 \\
        \alpha\sqrt{s_x^2+1}+\sqrt{t_x^2+(t_y+1)^2}           & \text{if } i = 9 \\
        \alpha\sqrt{(s_x-w_2)^2+1}+\sqrt{(t_x-w_2)^2+(t_y+1)^2}     & \text{if } i = 10
      \end{cases}\quad\left(~\parbox{0.15\textwidth}{and thus $t$ lies outside $R$}~\right)\text{, and}\\
      d_i(s,t) = &
      \begin{cases}
         s_x-t_x-\sqrt{\alpha^2-1}t_y   \hspace*{3.4cm} & \text{if } i = 11 \\
         \alpha\sqrt{(s_x-t_x)^2+t_y^2}  & \text{if } i = 12 \\
         \sqrt{\alpha^2-1}(s_x+t_x)-t_y  & \text{if } i = 13
      \end{cases}\quad\!\left(~\parbox{0.15\textwidth}{and thus $t$ lies inside $R$}~\right)\text{,}
    \end{align*}
    
    in which $w_1$ is the unique real solution in the interval $(t_y,0)$ to the equation
    \begin{equation*}
        \beta w_1^4-2t_y\beta w_1^3+\left[\alpha^2t_x^2+\beta t_y^2-s_x^2\right]w_1^2+2s_x^2t_yw_1-s_x^2t_y^2 = 0,
    \end{equation*}
    
    where $ \beta = \alpha^2-1$, and $w_2$ is the unique real solution in the interval $(t_x,s_x)$ to the equation
    \begin{equation}\label{eq:d9}
    \begin{split}
        & \beta w_2^4-2\beta(t_x+s_x)w_2^3+\left[\beta(s_x^2+t_x^2+4s_xt_x)+\alpha^2(1+t_y)^2-1\right]w_2^2\\
        &-2\left[\beta(t_xs_x^2+t_x^2s_x)+\alpha^2(1+t_y)^2s_x-t_x\right]w_2+\beta t_x^2s_x^2+\alpha^2(1+t_y)^2s_x^2-t_x^2=0.
    \end{split}
    \end{equation}
\end{theorem}

In the rest of this section we provide proofs for the lengths $d_i(s,t)$ for $i \in \{1,\dots,13\}$. The lengths for $i \in \{1, 3, 4, 9, 12\}$ are straightforward to compute and thus not included here. All angles in this section are measured relative to the vertical or horizontal line through a vertex of $ \pi_i(s,t) $ for a vertex on the top (and bottom) or left side of $R$, respectively. Recall that $\theta_c$ denotes the critical angle. The following two equations capture the properties we use of the critical angle. For $\alpha < 1$, we have
\begin{equation}
        \label{eq:critical2}
        \sin{\theta_{c}} = \alpha \Rightarrow \begin{cases}
            \cos{\theta_{c}} = \sqrt{1-\alpha^2} \\
            \tan{\theta_{c}} = \frac{\sin{\theta_{c}}}{\cos{\theta_{c}}} = \frac{\alpha}{\sqrt{1-\alpha^2}}.
        \end{cases}
    \end{equation}
    
And for $\alpha > 1$, we have
    \begin{equation}
        \label{eq:critical}
        \sin{\theta_c} = \frac{1}{\alpha} \Rightarrow \begin{cases}
            \cos{\theta_c} = \sqrt{1-\frac{1}{\alpha^2}} \\
            \tan{\theta_c} = \frac{\sin{\theta_c}}{\cos{\theta_c}} = \frac{\frac{1}{\alpha}}{\sqrt{1-\frac{1}{\alpha^2}}} = \frac{1}{\sqrt{\alpha^2-1}}.
        \end{cases}
    \end{equation} 

Note that we do not take into account the case where $ \alpha = 1 $ since in that case the shortest path from $ s $ to $ t $ is simply a straight-line segment.

We frequently use these equations in the rest of this section to determine the lengths of the paths $\pi_i(s,t)$, for all $ i $.

\begin{lemma}
    The length of $ \pi_2(s,t) $ is given by $ d_2(s,t) = \alpha(s_x-t_x)+\sqrt{1-\alpha^2}t_y $.
\end{lemma}

\begin{proof}
    Let $ \theta_{c} $ be the critical angle made by $ \pi_2(s,t) $ on the top boundary of $R$, and let $ (b,0) $ be the point where the shortest path leaves $R$. We use Equation (\ref{eq:critical2}) to obtain the value of $ b $: 
    \begin{equation}
        \label{eq:bending7}
        \frac{\alpha}{\sqrt{1-\alpha^2}} = \frac{\lvert b - t_x \rvert}{\lvert t_y \rvert} = \frac{b-t_x}{t_y} \Rightarrow b = t_x+\frac{\alpha}{\sqrt{1-\alpha^2}}t_y.
    \end{equation}
    
    We know that the weight of the shortest path $ \pi_{2}(s,t) $ is given by $ d_{2}(s,t) = \alpha|s_x-b| + \sqrt{(b-t_x)^2+t_y^2} $. By using Equation (\ref{eq:bending7}), we have that
    \begin{align*}
        d_{2}(s,t) & = \alpha\left(s_x-t_x-\frac{\alpha}{\sqrt{1-\alpha^2}}t_y\right) + \sqrt{\left(t_x+\frac{\alpha}{\sqrt{1-\alpha^2}}t_y-t_x\right)^2+t_y^2} \\
        & = \alpha\left(s_x-t_x-\frac{\alpha}{\sqrt{1-\alpha^2}}t_y\right) + \sqrt{\frac{\alpha^2t_y^2}{1-\alpha^2}+t_y^2}\\
        & = \alpha(s_x-t_x)-\frac{\alpha^2t_y}{\sqrt{1-\alpha^2}} + \sqrt{\frac{t_y^2}{1-\alpha^2}} = \alpha(s_x-t_x)-\frac{\alpha^2t_y}{\sqrt{1-\alpha^2}} + \frac{\lvert t_y \rvert}{\sqrt{1-\alpha^2}}\\
        & = \alpha(s_x-t_x)+\frac{1-\alpha^2}{\sqrt{1-\alpha^2}}t_y= \alpha(s_x-t_x)+\sqrt{1-\alpha^2}t_y.\forceqed
    \end{align*}
    \renewcommand{\qed}{}
\end{proof}

\begin{lemma}
    \label{lemma:bisector12}
    The length of $ \pi_5(s,t) $ is given by $ d_5(s,t) = s_x-\sqrt{2-\alpha^2}t_x-\sqrt{\alpha^2-1}t_y $.
\end{lemma}

\begin{proof}
    The shortest path $ \pi_5(s,t) $ from $s$ to $t$ intersects the top side of $R$, and then it enters~$R$ using the critical angle, see Figure~\ref{fig:bisector12}. We proceed to compute the coordinates of the vertices of the shortest path in this case.
    \begin{figure}[tb]
        \centering
        \includegraphics[page=2]{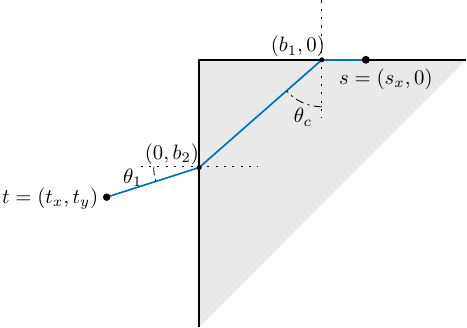}
        \caption{Illustration of the notation used for $\pi_5(s,t)$ in Lemma~\ref{lemma:bisector12}.}
        \label{fig:bisector12}
    \end{figure}
    
    Let $ \theta_1 $ be the angle at which the shortest path leaves $R$ with respect to the normal, see Figure~\ref{fig:bisector12}. Then $\sin{\theta_1} = \alpha\sin{\left(\frac{\pi}{2}-\theta_c\right)} = \alpha\cos{\theta_c} = \sqrt{\alpha^2-1}$, and thus
    \begin{equation}
        \label{eq:tantheta2}
        \tan{\theta_1} = \frac{\sin{\theta_1}}{\cos{\theta_1}} = \frac{\sqrt{\alpha^2-1}}{\sqrt{1-(\alpha^2-1)}} = \frac{\sqrt{\alpha^2-1}}{\sqrt{2-\alpha^2}}.
    \end{equation}
    
    Let $ (b_1,0) $ and $ (0,b_2) $ be, respectively, the points where the shortest path enters and leaves~$R$. We also know that $ \tan{\theta_1} = \frac{\lvert t_y - b_2 \rvert}{\lvert t_x \rvert} $. Since $ t_y < b_2 < 0 $, and $ t_x < 0 $, we use Equation (\ref{eq:tantheta2}) to obtain the value of~$ b_2 $: 
    \begin{equation}
        \label{eq:tantheta3}
        \frac{\sqrt{\alpha^2-1}}{\sqrt{2-\alpha^2}} = \frac{\lvert t_y - b_2 \rvert}{\lvert t_x \rvert} = \frac{b_2-t_y}{-t_x} \Rightarrow b_2 = t_y-\frac{\sqrt{\alpha^2-1}}{\sqrt{2-\alpha^2}}t_x.
    \end{equation}
    
    Also, since $ \tan{\theta_c} = \frac{\lvert b_1 \rvert}{\lvert b_2 \rvert} $, and $ b_2 < 0 < b_1 $, we use Equation (\ref{eq:critical}) to get the value of~$ b_1 $:
    \begin{equation*}
        \frac{1}{\sqrt{\alpha^2-1}} = \frac{\lvert b_1 \rvert}{\lvert b_2 \rvert} = \frac{b_1}{-b_2} \Rightarrow b_1 = -\frac{b_2}{\sqrt{\alpha^2-1}} = -\frac{t_y}{\sqrt{\alpha^2-1}}+ \frac{t_x}{\sqrt{2-\alpha^2}}.
    \end{equation*}
    
    Since $ \sin{\theta_c} = \frac{b_1}{\sqrt{b_1^2+b_2^2}} = \frac{1}{\alpha} \Rightarrow \sqrt{b_1^2+b_2^2} = b_1\alpha $. Thus: 
    \begin{align*}
        d_5(s,t) & = s_x - b_1 + \alpha\sqrt{b_1^2+b_2^2}+\sqrt{t_x^2+(t_y-b_2)^2} \\
        & = s_x-b_1 +b_1\alpha^2+\sqrt{t_x^2+\left(t_y-t_y+\frac{\sqrt{\alpha^2-1}}{\sqrt{2-\alpha^2}}t_x\right)^2} \\
        & = s_x + (\alpha^2-1)b_1+\sqrt{t_x^2+\frac{\alpha^2-1}{2-\alpha^2}t_x^2} = s_x + (\alpha^2-1)b_1+\lvert t_x \rvert\sqrt{1+\frac{\alpha^2-1}{2-\alpha^2}} \\
        & = s_x + (\alpha^2-1)b_1+\lvert t_x\rvert \sqrt{\frac{2-\alpha^2+\alpha^2-1}{2-\alpha^2}} \\
        & = s_x + (\alpha^2-1)\left(\frac{t_x}{\sqrt{2-\alpha^2}}-\frac{t_y}{\sqrt{\alpha^2-1}}\right) - t_x\frac{1}{\sqrt{2-\alpha^2}}\\
        &= s_x-\sqrt{\alpha^2-1}t_y+\frac{\alpha^2-1-1}{\sqrt{2-\alpha^2}}t_x = s_x-\sqrt{\alpha^2-1}t_y-\frac{2-\alpha^2}{\sqrt{2-\alpha^2}}t_x\\
        &= s_x-\sqrt{\alpha^2-1}t_y-\sqrt{2-\alpha^2}t_x.\forceqed
    \end{align*}
    \renewcommand{\qed}{}
\end{proof}

\begin{lemma}
    The length of $ \pi_6(s,t) $ is given by $ d_6(s,t) = \alpha\sqrt{s_x^2+w_1^2}+\sqrt{t_x^2+(t_y-w_1)^2} $, where $ w_1 $ is the unique real solution in the interval $ (t_y , 0) $ to the equation 
    \begin{equation*}
        (\alpha^2-1)w_1^4-2t_y(\alpha^2-1)w_1^3+\left[\alpha^2t_x^2+(\alpha^2-1)t_y^2-s_x^2\right]w_1^2+2s_x^2t_yw_1-s_x^2t_y^2 = 0.
    \end{equation*}
\end{lemma}

\begin{proof}
    Let $ (0,w_1) $ be the point where $ \pi_6(s,t) $ leaves $R$, and let $ \theta_1 $ and $ \theta_2 $ be, respectively, the angles of incidence an refraction at $(0,w_1)$. Then, by Snell's law of refraction, we get that $ \alpha \sin{\theta_1} = \sin{\theta_2} $. Thus,
    \begin{align*}
        \alpha\frac{|w_1|}{\sqrt{s_x^2+w_1^2}} & = \frac{|t_y-w_1|}{\sqrt{t_x^2+(t_y-w_1)^2}} \\
        \Rightarrow \alpha^2w_1^2\left(t_x^2+(t_y-w_1)^2\right) & =(t_y-w_1)^2(s_x^2+w_1^2)\\
        \Rightarrow \alpha^2w_1^2t_x^2+\alpha^2w_1^2t_y^2+\alpha^2w_1^4-2\alpha^2w_1^3t_y &=s_x^2t_y^2+s_x^2w_1^2-2s_x^2t_yw_1+w_1^2t_y^2+w_1^4-2w_1^3t_y.
    \end{align*}
    
    Hence,
    \begin{align*}
        (\alpha^2-1)w_1^4-2t_y(\alpha^2-1)w_1^3+\left[\alpha^2t_x^2+(\alpha^2-1)t_y^2-s_x^2\right]w_1^2+2s_x^2t_yw_1-s_x^2t_y^2 = 0.
    \end{align*}
    
    Finally, we get that the weighted length of the shortest path $ \pi_6(s,t) $ is given by
    \begin{equation*}
        d_6(s,t) = \alpha\sqrt{s_x^2+w_1^2}+\sqrt{t_x^2+(t_y-w_1)^2}.\forceqed
    \end{equation*}
    \renewcommand{\qed}{}
\end{proof}

\begin{lemma}
    The length of $ \pi_7(s,t) $ is given by $ d_7(s,t) = \sqrt{\alpha^2-1}s_x + 1 + \sqrt{t_x^2+(t_y+1)^2} $.
\end{lemma}

\begin{proof}
    Let $ (0,b_1) $ be the point where $ \pi_7(s,t) $ leaves $R$ for the first time, i.e., the first vertex of $ \pi_7(s,t) $. 
    Since $ b_1 < 0 $, and $ s_x > 0 $, we obtain the coordinates of this first vertex by using Equation (\ref{eq:critical}):
    \begin{equation}
        \label{eq:bendingpoint4}
        \tan{\theta_c} = \frac{\lvert b_1 \rvert}{\lvert s_x \rvert} = \frac{-b_1}{s_x} = \frac{1}{\sqrt{\alpha^2-1}} \Rightarrow b_1 = -\frac{s_x}{\sqrt{\alpha^2-1}}.
    \end{equation}
    
    The weight of the shortest path $ \pi_7(s,t) $ is then given by 
    \begin{align*}
        d_7(s,t) &= \alpha\sqrt{s_x^2+b_1^2} + (b_1+1) + \sqrt{t_x^2+(-1-t_y)^2}\\
        &= \frac{\alpha^2 s_x}{\sqrt{\alpha^2-1}} -\frac{s_x}{\sqrt{\alpha^2-1}} + 1 + \sqrt{t_x^2+(-1-t_y)^2}\\
        &= \sqrt{\alpha^2-1}s_x + 1 + \sqrt{t_x^2+(t_y+1)^2}.\forceqed
    \end{align*}
    \renewcommand{\qed}{}
\end{proof}

\begin{lemma}
    \label{lemma:length8}
    The length of $ \pi_8(s,t) $ is given by $ d_8(s,t) = \sqrt{\alpha^2-1}(s_x+t_x) - \sqrt{2-\alpha^2}(1+t_y)+1 $.
\end{lemma}

\begin{proof}
    Let $ (0,b_1) $ be the point where $ \pi_8(s,t) $ leaves $R$ for the first time and let $ (0,b_2) $ and $ (b_3,-1) $ be, respectively, the points where $ \pi_8(s,t) $ enters and leaves $R$ for the second time, see Figure~\ref{fig:bisector45}. As $\pi_7(s,t)$ and $\pi_8(s,t)$ overlap up to $b_2$, Equation~(\ref{eq:bendingpoint4}) gives us that $b_1 = -\frac{s_x}{\sqrt{\alpha^2-1}}$.
    \begin{figure}[tb]
        \centering
        \includegraphics[page=2]{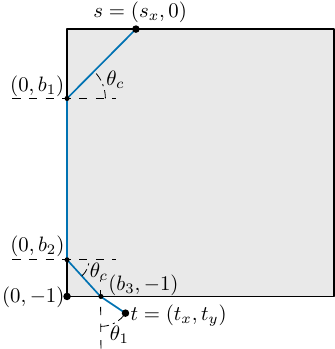}
        \caption{Illustration of the notation used for $\pi_8(s,t)$ in Lemma \ref{lemma:length8}.}
        \label{fig:bisector45}
      \end{figure}
      
     Recall that $R$ has height $ 1 $. Let $ \theta_1 $ be the angle at which the shortest path $ \pi_8(s,t) $ leaves~$R$ for the last time with respect to the normal, see Figure~\ref{fig:bisector45}. Then, $\tan{\theta_1} = \frac{\sqrt{\alpha^2-1}}{\sqrt{2-\alpha^2}}$, similar to Equation~(\ref{eq:tantheta2}). So,
    \begin{equation*}
        \tan{\theta_1} = \frac{\lvert t_x-b_3 \rvert}{\lvert t_y+1 \rvert} = \frac{t_x-b_3}{-1-t_y} = \frac{\sqrt{\alpha^2-1}}{\sqrt{2-\alpha^2}} \Rightarrow b_3 = t_x + (1+t_y)\frac{\sqrt{\alpha^2-1}}{\sqrt{2-\alpha^2}}.
        \end{equation*}
        
    And, using Equations (\ref{eq:critical}), we get
    \begin{equation*}
        \tan{\theta_c} = \frac{\lvert -1-b_2 \rvert}{\lvert b_3 \rvert} = \frac{b_2+1}{b_3} = \frac{1}{\sqrt{\alpha^2-1}} \Rightarrow b_2=\frac{b_3}{\sqrt{\alpha^2-1}}-1=\frac{t_x}{\sqrt{\alpha^2-1}}+\frac{1+t_y}{\sqrt{2-\alpha^2}}-1.
    \end{equation*}
        
    The weight of the shortest path $ \pi_8(s,t) $ is given by $ d_8(s,t) = \alpha\sqrt{s_x^2+b_1^2} + |b_2-b_1| + \alpha\sqrt{b_3^2+(b_2+1)^2} + \sqrt{(b_3-t_x)^2+(-1-t_y)^2} $. Using the expression for $b_1$, we have that
    \begin{equation}
        \label{eq:distancesz}
         \sqrt{s_x^2+b_1^2} = \sqrt{s_x^2+\left(-\frac{s_x}{\sqrt{\alpha^2-1}}\right)^2} = \sqrt{\frac{\alpha^2-1+1}{\alpha^2-1}s_x^2} = \frac{\alpha s_x}{\sqrt{\alpha^2-1}}.
    \end{equation}
    
    Using our expressions for $b_2$ and $b_3$, and the fact that $ b_2 + 1 > 0 $, we obtain the following for the second and third square root terms:
    \begin{align}
        \label{eq:distanceyx}
         \sqrt{b_3^2+(b_2+1)^2} &= \sqrt{\left(t_x+(1+t_y)\frac{\sqrt{\alpha^2-1}}{\sqrt{2-\alpha^2}}\right)^2+\left(\frac{t_x}{\sqrt{\alpha^2-1}}+\frac{1+t_y}{\sqrt{2-\alpha^2}}-1+1\right)^2}\nonumber\\
         &= \sqrt{(\alpha^2-1)\left(\frac{t_x}{\sqrt{\alpha^2-1}}+\frac{1+t_y}{\sqrt{2-\alpha^2}}\right)^2+\left(\frac{t_x}{\sqrt{\alpha^2-1}}+\frac{1+t_y}{\sqrt{2-\alpha^2}}\right)^2}\nonumber\\
         &=\sqrt{\alpha^2\left(\frac{t_x}{\sqrt{\alpha^2-1}}+\frac{1+t_y}{\sqrt{2-\alpha^2}}\right)^2} = \alpha\left(\frac{t_x}{\sqrt{\alpha^2-1}}+\frac{1+t_y}{\sqrt{2-\alpha^2}}\right),
    \end{align}
    
    and
    \begin{align}
        \label{eq:distancext}
        \sqrt{(b_3-t_x)^2+(-1-t_y)^2} &= \sqrt{\left(t_x+(1+t_y)\frac{\sqrt{\alpha^2-1}}{\sqrt{2-\alpha^2}}-t_x\right)^2+(-1-t_y)^2}\nonumber\\
        &= \sqrt{(1+t_y)^2\frac{\alpha^2-1}{2-\alpha^2}+(1+t_y)^2} = \sqrt{(1+t_y)^2\frac{\alpha^2-1+2-\alpha^2}{2-\alpha^2}}\nonumber\\
        &= \frac{\lvert 1+t_y \rvert}{\sqrt{2-\alpha^2}}.
    \end{align}
    
    Using Equations (\ref{eq:distancesz}), (\ref{eq:distanceyx}), and (\ref{eq:distancext}), we get that the weighted length of the path $ \pi_8(s,t) $ is given by
    \begin{align*}
        d_8(s,t) &= \frac{\alpha^2 s_x}{\sqrt{\alpha^2-1}} -\frac{s_x+t_x}{\sqrt{\alpha^2-1}}-\frac{1+t_y}{\sqrt{2-\alpha^2}} + 1 + \alpha^2\left(\frac{t_x}{\sqrt{\alpha^2-1}}+\frac{1+t_y}{\sqrt{2-\alpha^2}}\right) + \frac{\lvert 1+t_y \rvert}{\sqrt{2-\alpha^2}}\\
        &= \alpha^2\frac{s_x+t_x}{\sqrt{\alpha^2-1}} -\frac{s_x+t_x}{\sqrt{\alpha^2-1}}+(\alpha^2-1)\frac{1+t_y}{\sqrt{2-\alpha^2}} + 1 - \frac{1+t_y}{\sqrt{2-\alpha^2}}\\
        &= (\alpha^2-1)\frac{s_x+t_x}{\sqrt{\alpha^2-1}}+(\alpha^2-2)\frac{1+t_y}{\sqrt{2-\alpha^2}} + 1\\
        &= \sqrt{\alpha^2-1}(s_x+t_x) - \sqrt{2-\alpha^2}(1+t_y)+1.\forceqed
    \end{align*}
    \renewcommand{\qed}{}
\end{proof}

\begin{lemma}
   The length of $ \pi_{10}(s,t) $ is given by $ d_{10}(s,t) = \alpha\sqrt{(s_x-w_2)^2+1}+\sqrt{(t_x-w_2)^2+(t_y+1)^2} $ where $w_2$ is the unique real solution in the interval $ (t_x, s_x) $ to the equation 
    \begin{equation*}
    \begin{split}
        &(\alpha^2-1)w_2^4-2(\alpha^2-1)(t_x+s_x)w_2^3+\left[(\alpha^2-1)(s_x^2+t_x^2+4s_xt_x)+\alpha^2(1+t_y)^2-1\right]w_2^2\\
        &-2\left[(\alpha^2-1)(t_xs_x^2+t_x^2s_x)+\alpha^2(1+t_y)^2s_x-t_x\right]w_2+(\alpha^2-1)t_x^2s_x^2+\alpha^2(1+t_y)^2s_x^2-t_x^2=0.
    \end{split}
    \end{equation*}
\end{lemma}

\begin{proof}
    Let $ (w_2,-1) $ be the point where $ \pi_{10}(s,t) $ leaves $R$, and let $ \theta_1 $ and $ \theta_2 $ be, respectively, the angles of incidence and refraction at $(w_2,-1)$. Then, by Snell's law of refraction, we get that:
    \begin{align*}
        \alpha \sin{\theta_1} & = \sin{\theta_2} \Rightarrow \alpha\frac{s_x-w_2}{\sqrt{(s_x-w_2)^2+1}} = \frac{w_2-t_x}{\sqrt{(w_2-t_x)^2+(-1-t_y)^2}} \\
        & \Rightarrow \alpha^2(w_2-t_x)^2(s_x-w_2)^2+\alpha^2(-1-t_y)^2(s_x-w_2)^2=(w_2-t_x)^2(s_x-w_2)^2+(w_2-t_x)^2\\
        & \Rightarrow (\alpha^2-1)(w_2-t_x)^2(s_x-w_2)^2+\alpha^2(-1-t_y)^2(s_x-w_2)^2-(w_2-t_x)^2=0\\
        & \Rightarrow \left[(\alpha^2-1)w_2^2+(\alpha^2-1)t_x^2-2(\alpha^2-1)t_xw_2\right](s_x^2+w_2^2-2s_xw_2)+\alpha^2(-1-t_y)^2s_x^2\\
        & +\alpha^2(-1-t_y)^2w_2^2-2\alpha^2(-1-t_y)^2s_xw_2-w_2^2-t_x^2+2t_xw_2 = 0\\
        & \Rightarrow (\alpha^2-1)w_2^2s_x^2+(\alpha^2-1)w_2^4-2(\alpha^2-1)s_xw_2^3+(\alpha^2-1)t_x^2s_x^2+(\alpha^2-1)t_x^2w_2^2\\
        & -2(\alpha^2-1)t_x^2s_xw_2 -2(\alpha^2-1)t_xs_x^2w_2-2(\alpha^2-1)t_xw_2^3+4(\alpha^2-1)s_xt_xw_2^2\\
        & +\alpha^2(-1-t_y)^2s_x^2+\alpha^2(-1-t_y)^2w_2^2-2\alpha^2(-1-t_y)^2s_xw_2-w_2^2-t_x^2+2t_xw_2=0.
    \end{align*}
    
    Hence,
    \begin{align*}
        \label{eq:bendingpoint}
        &(\alpha^2-1)w_2^4-2(\alpha^2-1)(t_x+s_x)w_2^3+\left[(\alpha^2-1)(s_x^2+t_x^2+4s_xt_x)+\alpha^2(1+t_y)^2-1\right]w_2^2 \nonumber\\
        & -2\left[(\alpha^2-1)(t_xs_x^2+t_x^2s_x)+\alpha^2(1+t_y)^2s_x-t_x\right]w_2+(\alpha^2-1)t_x^2s_x^2+\alpha^2(1+t_y)^2s_x^2-t_x^2=0.
    \end{align*}
    
    Finally, we get that the weighted length of the shortest path $ \pi_{10}(s,t) $ is given by
    \begin{equation*}
        \label{eq:equalpaths}
        d_{10}(s,t) = \alpha\sqrt{(s_x-w_2)^2+1}+\sqrt{(t_x-w_2)^2+(1+t_y)^2}.\forceqed
    \end{equation*}
    \renewcommand{\qed}{}
\end{proof}

\begin{lemma}
    The length of $ \pi_{11}(s,t) $ is given by $ d_{11}(s,t) = s_x-t_x-\sqrt{\alpha^2-1}t_y $.
\end{lemma}

\begin{proof}
    Let $ (b_1,0) $ be the point where $ \pi_{11}(s,t) $ enters $R$. Let $ \theta_c $ be the angle at which $ \pi_{11}(s,t) $ enters $R$. Since  $ \theta_c $ is the critical angle, using Equation (\ref{eq:critical}), we get the value of $ b_1 $:
    \begin{equation}
        \label{eq:bendingpoint6}
        \tan{\theta_c} = \frac{\lvert b_1-t_x \rvert}{\lvert t_y \rvert} = \frac{b_1-t_x}{-t_y} = \frac{1}{\sqrt{\alpha^2-1}} \Rightarrow b_1 = t_x-\frac{t_y}{\sqrt{\alpha^2-1}}.
    \end{equation}
    
    We know that the weight of the shortest path $ \pi_{11}(s,t) $ is given by $ d_{11}(s,t) = |s_x-b_1| + \alpha\sqrt{(b_1-t_x)^2+t_y^2} $. By using Equation (\ref{eq:bendingpoint6}), we have that
    \begin{align*}
        d_{11}(s,t) & = \left(s_x-t_x+\frac{t_y}{\sqrt{\alpha^2-1}}\right) + \alpha\sqrt{\left(t_x-\frac{t_y}{\sqrt{\alpha^2-1}}-t_x\right)^2+t_y^2} \\
        & = \left(s_x-t_x+\frac{t_y}{\sqrt{\alpha^2-1}}\right) + \alpha\sqrt{\frac{t_y^2}{\alpha^2-1}+t_y^2}\\
        & = s_x-t_x+\frac{t_y}{\sqrt{\alpha^2-1}} + \alpha\sqrt{\frac{\alpha^2t_y^2}{\alpha^2-1}} = s_x-t_x+\frac{t_y}{\sqrt{\alpha^2-1}} + \frac{\alpha^2 \lvert t_y \rvert}{\sqrt{\alpha^2-1}}\\
        & = s_x-t_x+\frac{t_y}{\sqrt{\alpha^2-1}} - \frac{\alpha^2 t_y}{\sqrt{\alpha^2-1}} = s_x-t_x - \frac{(\alpha^2-1) t_y}{\sqrt{\alpha^2-1}}\\
        &= s_x-t_x-\sqrt{\alpha^2-1}t_y.\forceqed
    \end{align*}
    \renewcommand{\qed}{}
\end{proof}
\clearpage
\begin{lemma}
    \label{lem:path12}
    The length of $ \pi_{13}(s,t) $ is given by $ d_{13}(s,t) = \sqrt{\alpha^2-1}(s_x+t_x)-t_y $.
\end{lemma}

\begin{proof}
    Let $ (0,b_1) $ and $ (0,b_2)$ be the points where $ \pi_{13}(s,t) $ leaves and enters for the second time, respectively, the region $R$. From Lemma \ref{lemma:length8} we know that $ b_1=-\frac{s_x}{\sqrt{\alpha^2-1}}$. Using Equation~(\ref{eq:critical}), we find:
    \begin{equation*}
        \tan{\theta_c} = \frac{\lvert t_y-b_2 \rvert}{\lvert t_x \rvert} = \frac{b_2-t_y}{t_x} = \frac{1}{\sqrt{\alpha^2-1}} \Rightarrow b_2=\frac{t_x}{\sqrt{\alpha^2-1}}+t_y.
    \end{equation*}
    
    We then get that the weight of $ \pi_{13}(s,t) $ is given by $ d_{13}(s,t) = \alpha\sqrt{s_x^2+b_1^2} + |b_2-b_1| + \alpha\sqrt{t_x^2+(b_2-t_y)^2} $. So:
    \begin{align*}
        d_{13}(s,t) &= \alpha\sqrt{s_x^2+\frac{s_x^2}{\alpha^2-1}}-\frac{s_x}{\sqrt{\alpha^2-1}}-\frac{t_x}{\sqrt{\alpha^2-1}}-t_y+\alpha\sqrt{t_x^2+\left(\frac{t_x}{\sqrt{\alpha^2-1}}+t_y-t_y\right)^2} \\
        = & \alpha\sqrt{\frac{\alpha^2s_x^2}{\alpha^2-1}}-\frac{s_x+t_x}{\sqrt{\alpha^2-1}}-t_y+\alpha\sqrt{t_x^2+\frac{t_x^2}{\alpha^2-1}} = \frac{\alpha^2\lvert s_x \rvert}{\sqrt{\alpha^2-1}}-\frac{s_x+t_x}{\sqrt{\alpha^2-1}}-t_y+\alpha\sqrt{\frac{\alpha^2t_x^2}{\alpha^2-1}}\\
        = & \frac{\alpha^2 s_x}{\sqrt{\alpha^2-1}}-\frac{s_x+t_x}{\sqrt{\alpha^2-1}}-t_y+\frac{\alpha^2t_x}{\sqrt{\alpha^2-1}} = \frac{\alpha^2(s_x+t_x)}{\sqrt{\alpha^2-1}}-\frac{s_x+t_x}{\sqrt{\alpha^2-1}}-t_y\\
        =& \frac{\alpha^2-1}{\sqrt{\alpha^2-1}}(s_x+t_x)-t_y = \sqrt{\alpha^2-1}(s_x+t_x)-t_y.\forceqed
    \end{align*}
    \renewcommand{\qed}{}
\end{proof}

\subsection{The source point $s$ lies inside $R$}
\label{sub:The_source_point_lies_inside}

We now consider the case where $ s $ is restricted to the interior of the rectangle $ R $. 
\begin{observation}
      Let $s$ be a point in a rectangle $R$ with weight $ 0 < \alpha < \sqrt{2} $. There are $ 6 $ types of shortest paths, up to symmetries, namely $ \pi_i(s,t) $, for $ i \in \{6, 7, 8, 9, 10, 12, 13\} $. 
\end{observation}

The types of shortest paths are similar to the ones defined in Observation~\ref{obs:typesboundary}, see the paths in Figure~\ref{fig:on_boundary} where the top side of $R$ or the region above $R$ is not intersected. As in Theorem~\ref{thm:length}, we can thus compute (the length of) a shortest path (of each type) exactly, albeit that the expressions for the length are dependent on the location of $s$ in $R$. Note that Theorem~\ref{thm:length} gives exact lengths for all path types when $R$ has height $>1$ and $s$ is at distance exactly~1 from the bottom boundary of $R$.

\subsection{The source point $s$ lies outside of $R$}
\label{sub:source_point_outside}

When both the source and the target point are outside of $R$, the shortest path can again be of many different types. In particular, the types in Figure~\ref{fig:on_boundary} can be generalized to this setting. There are two special cases where the shortest path bends \emph{twice}, and these two vertices do not have the critical angle property: it can bend on two opposite sides of the rectangle, or on two incident sides. In the first case, the angles at both vertices are equal, and the shortest path can be computed exactly~\cite{narayanappa2006exact}. For the second case, we show that it is not possible to compute the coordinates of the vertices exactly in the \ACMQ. Hence, the WRP limited to two weights $\{1,\alpha\}$ is not solvable within the \ACMQ. Note that this path type can occur in an even simpler setting, where $R$ is a single quadrant instead of a rectangle.

\begin{theorem}
    \label{thm:nonsolvable}
    The Weighted Region Problem with weights in the set
    $ \{1, \alpha\} $,  with
    $ 0 < \alpha < \sqrt{2} $, and $ \alpha \neq 1 $, cannot be solved exactly within the
    \ACMQ, even if $R$ is a single quadrant.
\end{theorem}

\begin{proof}
Consider the situation where a horizontal and a vertical line intersect at the point $ O = (50,150) $. Let $ R $ be the quadrant such that $ O $ is its top-left corner, and has weight $ \alpha = 1.2 $. Recall that the weight outside $ R $ is $ 1 $. Let $ s=(0,0) $ be the source point and $ t = (200,200) $ be the target point, see Figure~\ref{fig:notation_polynomial_exact}. We follow the approach of~\cite{carufel14} to show that the polynomial that represents a solution to the Weighted Region Problem in this situation is not solvable within the \ACMQ. The following lemma, which is a consequence of Theorem $ 1 $ and Lemma $ 2 $ of~\cite{carufel14}, see also~\cite{bajaj1988algebraic,dummit2004abstract}, states when a polynomial is unsolvable within the \ACMQ.
\begin{figure}
    \centering
    \includegraphics[page=2]{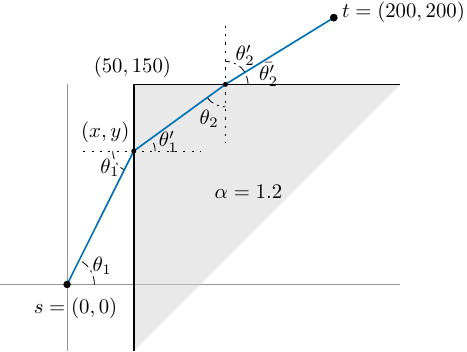}
    \caption{The set-up for the proof that even for two weights the Weighted Region Problem cannot be solved within the \ACMQ.}
    \label{fig:notation_polynomial_exact}
\end{figure}

\begin{lemma}
    \label{lem:unsolvable}
    Let $p(x)$ be a polynomial of odd degree $d \geq 5$. Suppose there are three primes $q_1,q_2,q_3$ that do not divide the discriminant of $p(x)$, such that 
    \begin{align*}
        &p(x) \equiv p_d(x) \mod q_1,\\
        &p(x) \equiv p_1(x)p_{d-1}(x) \mod q_2, \text{ and }\\
        &p(x) \equiv p_2(x)p_{d-2}(x) \mod q_3,
    \end{align*}
    where $p_i(x)$ denotes an irreducible polynomial of degree $i$ modulo the given prime. Then $p(x) = 0$ is unsolvable within the \ACMQ.
\end{lemma}

 Let $(x,y)$ be the coordinates of the intersection point of the path $\pi(s,t)$ with the vertical side of the quadrant. We denote by $ \theta_1 $ the angle made by the ray from $s$ to $(x,y)$ with the perpendicular to the vertical side of the quadrant, by $ \theta'_1 $ the angle of the refracted ray with respect to the same line, and by $ \bar{\theta'_2} $ the angle of the refracted ray with respect to the top side of the quadrant, see Figure~\ref{fig:notation_polynomial_exact}. We can then express the sum of the horizontal distances in terms of tangents of the angles, as follows:
\begin{equation*}
    50 + \frac{150 - y}{\tan \theta'_1} + \frac{50}{\tan \bar{\theta'_2}} = 200 \Longrightarrow 0 = \frac{150 - y}{\tan \theta'_1} + \frac{50}{\tan \bar{\theta'_2}} -150.
\end{equation*}

Using that $y = 50 \tan \theta_1$, we obtain an equation only containing $\theta_1, \theta'_1$ and $ \bar{\theta'_2}$.
\begin{equation*}
    0 = \frac{150 - 50 \tan\theta_1}{\tan \theta'_1} + \frac{50}{\tan \bar{\theta'_2}} -150.
\end{equation*}

We then apply the trigonometric identities $\tan \theta = \frac{\sin \theta}{\sqrt{1 - \sin^2\theta}}$, for $ \theta = \theta_1 $ and $\theta = \theta'_1$, and $\tan \theta = \frac{\sqrt{1 - \cos^2\theta}}{\cos \theta}$, for $\theta = \bar{\theta'_2}$.
\begin{equation}\label{eq:expression_sin123}
    0 = \frac{150 - 50 \frac{\sin \theta_1}{\sqrt{1 - \sin^2\theta_1}}}{\frac{\sin \theta'_1}{\sqrt{1 - \sin^2\theta'_1}}} + \frac{50}{\frac{\sqrt{1 - \cos^2\bar{\theta'_2}}}{\cos \bar{\theta'_2}}} -150.
\end{equation}

According to Snell's law we have that $ \sin \theta'_1 = \sin \theta_1/\alpha$. Furthermore, together with the fact that $\cos \bar{\theta'_2} = \alpha \cos \theta'_1$, and $\cos \theta'_1 = \sqrt{1-\sin^2{\theta'_1}}$, we can express $\cos \bar{\theta'_2}$ as $ \cos \bar{\theta'_2} = \sqrt{\alpha^2 - \sin^2\theta_1}$. We replace all instances of $\sin \theta'_1$ and $\cos \bar{\theta'_2}$ in Equation~(\ref{eq:expression_sin123}) by these expressions.
\begin{align*}
    0 &= \left(150 - 50 \frac{\sin \theta_1}{\sqrt{1 - \sin^2\theta_1}}\right)\cdot \frac{\sqrt{\alpha^2 - \sin^2 \theta_1}}{\sin \theta_1} + \frac{50\sqrt{\alpha^2 - \sin^2\theta_1}}{\sqrt{1 - (\alpha^2 - \sin^2\theta_1)}}-150\\
     &= 50\sqrt{\alpha^2 - \sin^2 \theta_1}  \left(\frac{3}{\sin\theta_1} -  \frac{1}{\sqrt{1 - \sin^2\theta_1}}+ \frac{1}{\sqrt{1 - \alpha^2 + \sin^2\theta_1}}\right)-150.
\end{align*}

The final equation in terms of $u = \sin \theta_1$ then becomes
\begin{equation*}
    \sqrt{\alpha^2 - u^2}\left( \frac{3}{u} - \frac{1}{\sqrt{1-u^2}} + \frac{1}{\sqrt{1-\alpha^2 + u^2}}\right) = 3.
\end{equation*}

For $\alpha = 1.2$, this can be transformed into the following polynomial by squaring appropriately using Mathematica~\cite{Mathematica}:
\begin{align*}
    p(u) = &-5602195930320001+93511401766200000 u-713160370741499900 u^2\\
    &+3259398736514250000 u^3
    -9869397269940000000 u^4+20717559301050000000 u^5\\
    &-30701172521250000000 u^6+32082903984375000000 u^7-23159988281250000000 u^8\\
    &+10999072265625000000 u^9-3093750000000000000 u^{10} + 390625000000000000 u^{11}.
\end{align*}

    To show that the polynomial $ p(u) $ is unsolvable, we thus need three primes $q_1,q_2,q_3$ that adhere to the conditions in Lemma~\ref{lem:unsolvable}. Using Mathematica we find the following expressions for $p(u)$ modulo $59$, $37$, and $17$, respectively:
    \begin{align*}
        & 46 \left(u^{11}+44 u^{10}+32 u^9+33 u^8+26 u^7+47 u^6+21 u^5+11 u^4+38 u^3+3 u^2+6 u+42\right),\\
        & 16 (u+17) \left(u^{10}+18 u^9+23 u^8+23 u^7+35 u^6+8 u^5+34 u^4+16 u^3+11 u^2+34 u+10\right),\\
        & 4 \left(u^2+14 u+9\right) \left(u^9+8 u^8+11 u^7+3 u^6+5 u^5+2 u^4+2 u^3+12 u^2+9 u+16\right).
    \end{align*}

    We conclude that even the very limited Weighted Region Problem where we allow for a single quadrant to have positive weight unequal to $1$ and $s$ and $t$ are on halfplanes bounded by the sides of the quadrant, not containing the quadrant, is not solvable within the \ACMQ.
\end{proof}

\section{Computing a Shortest Path Map}
\label{sec:Computing_a_shortest_path_map}

To find a shortest path from a source point $s$ to \emph{all} points at once, one can build a \emph{Shortest Path Map} (\SPM), see e.g.,~\cite{hershberger1999optimal,mitchell1993shortest,mitchell1991weighted}. A \SPM is a subdivision of the space for a given source~$s$, where for each cell the paths $\pi(s,t)$, with $t$ in the cell, have the same type. With it, we are able to find for each specific destination $t$, the weight of the shortest path from $s$ to $t$ simply by locating the point $t$ in the subdivision. Once a \SPM is available, we are able to report weights of shortest paths from $s$ to any destination~$t$ in $O(\log {n})$ time by standard point location techniques~\cite{edelsbrunner1986optimal,kirkpatrick1983optimal}. To compute the \SPM, we consider computing the bisectors $b_{i,j} = \{ q \mid q \in \mathbb{R}^2 \land d_i(s,q) = d_j(s,q) \}$ for all relevant pairs of shortest path types $\pi_i, \pi_j$, i.e., pairs for which $b_{i,j}$ appears in the Shortest Path Map.

A \SPM requires only polynomial space. However, in general, the bisector curves that bound cells of the \SPM subdivision will be curves of very high degree~\cite{carufel14,hershberger2022near}. As before, we consider the setting where  $R$ is a rectangular region. In Section~\ref{sub:spm_The_source_point_lies_on_the_boundary}, we first consider the case when $s$ lies on the boundary of~$R$. In Section~\ref{sub:spm_The_source_point_lies_inside}, we do the same for the case where $s$ lies inside~$R$. The case where $s$ lies outside~$R$ is
not interesting, as we cannot even compute exactly a single shortest path in that case.

\subsection{The source point $s$ lies on the boundary of $R$}
\label{sub:spm_The_source_point_lies_on_the_boundary}

The \SPM is given by the boundary of $R$ and several bisector curves, expressed as points $(x, b_{i,j}(x))$. If $\alpha < 1$, these curves all lie outside $R$ (the interior of $R$ is a single region in the \SPM). Furthermore, most of the bisectors involving $\pi_{10}(s,t)$ are of a much more complicated form, as might be expected from the implicit representation used for $d_{10}(s,t)$ in Theorem~\ref{thm:length}. Therefore, Lemma~\ref{lemma:bisectorsvisible} gives the bisector curves, excluding most of those related to $\pi_{10}(s,t)$.

\begin{lemma}
\label{lemma:bisectorsvisible}
    The \SPM for a point $ s=(s_x,0) $ on the boundary of the region $R$ is defined by:
    \begin{align*}
      b_{i,j}(x) = &
      \begin{cases}
        \frac{\sqrt{1-\alpha^2}}{\alpha}(s_x-x) & \text{if } i = 1, j = 2 \\
            -\frac{\sqrt{1-\alpha^2}}{\alpha}x \hspace*{1.8cm} & \text{if } i = 2, j = 3\\
            0 & \text{if } i = 3, j = 6\\
            0 & \text{if } i = 3, j = 9\\
            -1 +\frac{\alpha}{\sqrt{s_x^2+1-\alpha^2}}x & \text{if } i = 6, j = 9\\
            -1+\frac{\sqrt{1-(\alpha^2-1) s_x^2}}{\alpha s_x}x & \text{if } i = 9, j = 10
      \end{cases}\qquad\!\left(~
                   \parbox{0.13\textwidth}{when $\alpha < 1$}\right)\text{, and} \\
      b_{i,j}(x) = &
      \begin{cases}
         0 & \text{if } i = 1, j = 4 \\
            \frac{\sqrt{\alpha^2-1}}{\sqrt{2-\alpha^2}}x & \text{if } i = 4, j = 5 \\
            \frac{\sqrt{\alpha^2-1}}{\sqrt{2-\alpha^2}}x -\sqrt{\alpha^2-1}s_x & \text{if } i = 5, j = 6 \\
            x = 0 & \text{if } i = 6, j = 7 \\
            -1 -\frac{\sqrt{2-\alpha^2}}{\sqrt{\alpha^2-1}}x & \text{if } i = 7, j = 8\\
            -\sqrt{\alpha^2-1}(s_x-x) & \text{if } i = 11, j = 12\\
            - \frac{(s_x+x)+2\alpha\sqrt{s_x x}}{\sqrt{\alpha^2-1}} & \text{if } i = 12, j = 13
      \end{cases}\quad\!\left(~
                   \parbox{0.2\textwidth}{when $1 < \alpha < \sqrt{2}$}\right)\text{.}
    \end{align*}
\end{lemma}

We conjecture the following on some of the bisectors involving $ \pi_{10}(s,t) $.
\begin{conjecture}
\label{conj:pi9}
    No point on $b_{i,10}(x) \setminus R, \ i \in \{4,\ldots,8\}$, can be computed exactly within \ACMQ.
\end{conjecture}

We tried to prove this conjecture by taking a similar approach as in Theorem~\ref{thm:nonsolvable}. However, the solution to Equation~(\ref{eq:d9}) already seems to be of high degree. We therefore did not manage to formulate a point on the bisector as a polynomial equation (not containing roots).

Note that in the more restrictive case where $R$ is a single quadrant and $s$ lies on its boundary, the only types of shortest paths that exist are $ \pi_i(s,t) $, for $ i \in \{1,2,3,4,5,6,11,12,13\} $. Thus, we can compute the \SPM in the \ACMQ (the bisectors are given by some of the equations in Lemma~\ref{lemma:bisectorsvisible}).

Next, we provide the proofs for the bisector curves given in Lemma~\ref{lemma:bisectorsvisible}. We express each bisector as an explicit function of the shape $y=f(x)$. Then, the actual bisector is given by the corresponding points $ (x,y) $ Recall that $s = (s_x, 0)$ and $t = (t_x,t_y)$.

\begin{lemma}
    The bisector $ b_{1,2} $ is given by $ y = \frac{\sqrt{1-\alpha^2}}{\alpha}(s_x-x) $.
\end{lemma}

\begin{proof}
    We want to compute the coordinates of the points such that the weighted length of paths $ \pi_1(s, t) $ and $ \pi_2(s, t) $ is the same, i.e., the points $ (t_x, t_y) $ such that $ \sqrt{(s_x-t_x)^2+t_y^2} = \alpha(s_x-t_x)+\sqrt{1-\alpha^2}t_y $. Thus:
    \begin{align*}
        (s_x-t_x)^2+t_y^2&=\alpha^2(s_x-t_x)^2+(1-\alpha^2)t_y^2+2\alpha\sqrt{1-\alpha^2}(s_x-t_x)t_y\\
        0 &= (1-\alpha^2)(s_x-t_x)^2+\alpha^2t_y^2-2\alpha\sqrt{1-\alpha^2}(s_x-t_x)t_y\\
        0 &= \left[\sqrt{1-\alpha^2}(s_x-t_x)-\alpha t_y\right]^2\\
        \sqrt{1-\alpha^2}(s_x-t_x)&=\alpha t_y \Rightarrow t_y = \frac{\sqrt{1-\alpha^2}}{\alpha}(s_x-t_x).\forceqed
    \end{align*}
    \renewcommand{\qed}{}
\end{proof}

\begin{lemma}
    The bisector $ b_{2,3} $ is given by $ y = -\frac{\sqrt{1-\alpha^2}}{\alpha}x $.
\end{lemma}

\begin{proof}
    We want to compute the coordinates of the points such that the weighted length of paths $ \pi_2(s, t) $ and $ \pi_3(s, t) $ is the same, i.e., the points $ (t_x, t_y) $ such that  $ \alpha(s_x-t_x)+\sqrt{1-\alpha^2}t_y = \alpha s_x + \sqrt{t_x^2+t_y^2} $. Thus:
    \begin{align*}
        -\alpha t_x+\sqrt{1-\alpha^2}t_y &= \sqrt{t_x^2+t_y^2}\\
        \alpha^2 t_x^2+(1-\alpha^2)t_y^2-2\alpha\sqrt{1-\alpha^2}t_xt_y &= t_x^2+t_y^2 \\
        (1-\alpha^2)t_x^2+\alpha^2t_y^2+2\alpha\sqrt{1-\alpha^2}t_xt_y &= 0\\
        \left[\sqrt{1-\alpha^2}t_x+\alpha t_y\right]^2 &= 0\\
        \sqrt{1-\alpha^2}t_x=-\alpha t_y &\Rightarrow t_y = -\frac{\sqrt{1-\alpha^2}}{\alpha}t_x.\forceqed
    \end{align*}
    \renewcommand{\qed}{}
\end{proof}
\clearpage
\begin{lemma}
    \label{lemma:bisector45}
    The bisector $ b_{4,5} $ is given by $ y = \frac{\sqrt{\alpha^2-1}}{\sqrt{2-\alpha^2}}x $.
\end{lemma}

\begin{proof}
    We want to compute the coordinates of the points such that the weighted length of paths $ \pi_4(s, t) $ and $ \pi_5(s, t) $ is the same, i.e., the points $ (t_x, t_y) $ such that $ s_x + \sqrt{t_x^2+t_y^2} = s_x - \sqrt{2-\alpha^2}t_x-\sqrt{\alpha^2-1}t_y $. Thus:
    \begin{align*}
        \sqrt{t_x^2+t_y^2} & = -\sqrt{\alpha^2-1}t_y-\sqrt{2-\alpha^2}t_x \\
         t_x^2+t_y^2 & = (\alpha^2-1)t_y^2+(2-\alpha^2)t_x^2+2\sqrt{\alpha^2-1}\sqrt{2-\alpha^2}t_xt_y \\
          0&=(2-\alpha^2)t_y^2-2\sqrt{\alpha^2-1}\sqrt{2-\alpha^2}t_xt_y+(\alpha^2-1)t_x^2\\
        0&=\left[\sqrt{2-\alpha^2}t_y-\sqrt{\alpha^2-1}t_x\right]^2  \\
        \sqrt{2-\alpha^2}t_y &= \sqrt{\alpha^2-1}t_x 
         \Rightarrow t_y = \frac{\sqrt{\alpha^2-1}}{\sqrt{2-\alpha^2}}t_x.\forceqed
    \end{align*}
    \renewcommand{\qed}{}
\end{proof}

\begin{lemma}
    The bisector $ b_{5,6} $ is given by $ y = \frac{\sqrt{\alpha^2-1}}{\sqrt{2-\alpha^2}}x -\sqrt{\alpha^2-1}s_x$.
\end{lemma}

\begin{proof}
    The path $ \pi_6(s,t) $ is a special case of $ \pi_5(s,t) $ where the point of entry in $R$ is $s$. The path $ \pi_6(s,t) $ does not have the critical angle property. However, the points on the bisector, still have that critical angle property, since the shortest path from $s$ to them is of both types. Let $ (0,b_2) $ be the point where $ \pi_5(s,t) $ leaves the square, see Figure~\ref{fig:bisector12}. Using Equation (\ref{eq:critical}) we obtain the following relation:
    \begin{equation*}
        \tan{\theta_c} = \frac{\lvert s_x \rvert}{\lvert b_2 \rvert} \Rightarrow \lvert b_2 \rvert = \frac{\lvert s_x \rvert}{\tan{\theta_c}} = \sqrt{\alpha^2-1}\lvert s_x \rvert = \sqrt{\alpha^2-1}s_x \Rightarrow b_2 = -\sqrt{\alpha^2-1}s_x.
    \end{equation*}
    
    For $\pi_5(s,t)$, we obtained $b_2 = t_y-\frac{\sqrt{\alpha^2-1}}{\sqrt{2-\alpha^2}}t_x$ in Lemma~\ref{lemma:bisector12} (see Equation~(\ref{eq:tantheta3})). We then obtain the equation of the bisector $ b_{5,6} $:
    \begin{equation*}
        t_y = b_2 + \frac{\sqrt{\alpha^2-1}}{\sqrt{2-\alpha^2}}t_x = \frac{\sqrt{\alpha^2-1}}{\sqrt{2-\alpha^2}}t_x -\sqrt{\alpha^2-1}s_x.\forceqed
    \end{equation*}
    \renewcommand{\qed}{}
\end{proof}

\begin{lemma}
    The bisector $ b_{6,9} $ is given by $ y = 1 +\frac{\alpha}{\sqrt{s_x^2+1-\alpha^2}}x $.
\end{lemma}
\begin{proof}
    We want the curve defined by the points $ (t_x,t_y) $ such that $ d_6(s,t) = d_9(s,t) $. The lengths $ d_6(s,t) $ and $ d_9(s,t) $ have the same length when $ w_1=-1 $. Hence, we want to compute the points $ (t_x,t_y) $ such that $ \beta+2t_y\beta+\alpha^2t_x^2+\beta t_y^2-s_x^2-2s_x^2t_y-s_x^2t_y^2 = 0 $, where $ \beta = \alpha^2-1$. Thus,
    \begin{align*}
        0 & = (\beta-s_x^2)t_y^2+2(\beta-s_x^2)t_y+\beta-s_x^2+\alpha^2t_x^2 \\
        \Rightarrow t_y & = \frac{-2(\beta-s_x^2)\pm\sqrt{4(\beta-s_x^2)^2-4(\beta-s_x^2)\left((\beta-s_x^2)+\alpha^2t_x^2\right)}}{2(\beta-s_x^2)}\\
        & = \frac{-2(\beta-s_x^2)\pm\sqrt{4(s_x^2-\beta)\alpha^2t_x^2}}{2(\beta-s_x^2)} = -1\pm\frac{2\alpha |t_x|\sqrt{s_x^2-\beta}}{2(\beta-s_x^2)} = -1\pm\frac{\alpha |t_x|}{\sqrt{s_x^2-\beta}}.
    \end{align*}
    
    By the way the path $ \pi_9(s,t) $ is defined, $ t_x \leq 0 $ and $ t_y \leq -1 $. Thus, the bisector is given by the curve $ -1+\frac{\alpha t_x}{\sqrt{s_x^2-\beta}} $.
\end{proof}

\begin{lemma}
    The bisector $ b_{7,8} $ is given by $ y = -1 -\frac{\sqrt{2-\alpha^2}}{\sqrt{\alpha^2-1}}x $.
\end{lemma}

\begin{proof}
    We want to know the coordinates of the points $ (t_x, t_y) $ such that $ \sqrt{\alpha^2-1}(s_x+t_x) - \sqrt{2-\alpha^2}(1+t_y)+1 = \sqrt{\alpha^2-1}s_x + 1 + \sqrt{t_x^2+(t_y+1)^2} $. Thus,
    \begin{align*}
        & \sqrt{\alpha^2-1}t_x-\sqrt{2-\alpha^2}(1+t_y) = \sqrt{t_x^2+(1+t_y)^2}\\
        \Rightarrow &(\alpha^2-1)t_x^2+(2-\alpha^2)(1+t_y)^2-2\sqrt{\alpha^2-1}\sqrt{2-\alpha^2}(1+t_y)t_x = t_x^2+(1+t_y)^2\\
        \Rightarrow &(2-\alpha^2)t_x^2+(\alpha^2-1)(1+t_y)^2+2\sqrt{2-\alpha^2}\sqrt{\alpha^2-1}(1+t_y)t_x=0\\
        \Rightarrow &\left[\sqrt{2-\alpha^2}t_x+\sqrt{\alpha^2-1}(1+t_y)\right]^2 = 0 \Rightarrow \sqrt{2-\alpha^2}t_x=-\sqrt{\alpha^2-1}(1+t_y)\\
        \Rightarrow &t_y = -1-\frac{\sqrt{2-\alpha^2}}{\sqrt{\alpha^2-1}}t_x.\forceqed
    \end{align*}
    \renewcommand{\qed}{}
\end{proof}

\begin{lemma}
    The bisector $ b_{9,10} $ is given by $ y = \frac{\sqrt{1-(\alpha^2-1) s_x^2}}{\alpha s_x}x-1 $.
\end{lemma}

\begin{proof}
    We want the curve defined by the points $ (t_x,t_y) $ such that $ d_9(s,t) = d_{10}(s,t) $. The lengths $ d_9(s,t) $ and $ d_{10}(s,t) $ have the same length when $ w_2=0 $. Hence, we want to compute the points $ (t_x,t_y) $ such that $ \beta t_x^2s_x^2+\alpha^2(1+t_y)^2s_x^2-t_x^2 = 0 $, where $ \beta = \alpha^2-1$. Thus,
    \begin{align*}
        0 & = \alpha^2(1+t_y)^2s_x^2-t_x^2+\beta t_x^2s_x^2 \\
        \Rightarrow t_y & = \pm\frac{\sqrt{t_x^2(1-\beta s_x^2)}}{\alpha s_x}-1 = \pm\frac{\sqrt{1-\beta s_x^2}}{\alpha s_x}|t_x|-1.
    \end{align*}
    
    By the way the paths $ \pi_9(s,t) $ and $ \pi_{10}(s,t) $ are defined, $ t_y < -1 $. Thus, the bisector is given by the curve $ \frac{\sqrt{1-\beta s_x^2}}{\alpha s_x}t_x-1 $.
\end{proof}

\begin{lemma}
    \label{lemma:bisector78}
    The bisector $ b_{11,12} $ is given by $ y = -\sqrt{\alpha^2-1}(s_x-x) $.
\end{lemma}

\begin{proof}
    We want the curve defining the bisector between the region containing the points $ (t_x, t_y) $ such that the shortest path from $ s=(s_x,0) $ to $ t $ is $ \pi_{11}(s,t) $, and the region containing the points $ t=(t_x,t_y) $ such that the shortest path from $ s=(s_x,0) $ to $ t$ is $ \pi_{12}(s,t) $, i.e., the points $ (t_x,t_y) $ such that $ \alpha\sqrt{(s_x-t_x)^2+t_y^2} = s_x-t_x-\sqrt{\alpha^2-1}t_y $. Thus,
    \begin{align*}
        &\alpha^2[(s_x-t_x)^2+t_y^2] = (s_x-t_x)^2+(\alpha^2-1)t_y^2-2(s_x-t_x)\sqrt{\alpha^2-1}t_y\\
        &\Rightarrow 0 = \alpha^2(s_x-t_x)^2-(s_x-t_x)^2+\alpha^2t_y^2-\alpha^2t_y^2+t_y^2+2(s_x-t_x)\sqrt{\alpha^2-1}t_y\\
         &\hspace{0.7cm}= (\alpha^2-1)(s_x-t_x)^2+t_y^2+2(s_x-t_x)^2\sqrt{\alpha^2-1}t_y = (t_y+\sqrt{\alpha^2-1}(s_x-t_x))^2\\
         &\Rightarrow t_y = -\sqrt{\alpha^2-1}(s_x-t_x).\forceqed
    \end{align*}
    \renewcommand{\qed}{}
\end{proof}

\begin{lemma}
    The bisector $ b_{12,13} $ bisector is given by $ y = - \frac{(s_x+x)+2\alpha\sqrt{s_xx}}{\sqrt{\alpha^2-1}} $.
\end{lemma}

\begin{proof}
    We want the curve defining the bisector between the region containing the points $ (t_x, t_y) $ such that the shortest path from $ s=(s_x,0) $ to $ t $ is $ \pi_{12}(s,t) $, and the region containing the points $ t=(t_x,t_y) $ such that the shortest path from $ s=(s_x,0) $ to $ t$ is $ \pi_{13}(s,t) $, i.e., the points $ (t_x,t_y) $ such that $\sqrt{\alpha^2-1}(s_x+t_x)-t_y = \alpha\sqrt{(s_x-t_x)^2+t_y^2}$. Thus,
    \begin{align*}
        &[\sqrt{\alpha^2-1}(s_x+t_x)-t_y]^2 = \alpha^2((s_x-t_x)^2+t_y^2)\\
        &\Rightarrow (\alpha^2-1)(s_x+t_x)^2+t_y^2-2\sqrt{\alpha^2-1}(s_x+t_x)t_y = \alpha^2(s_x-t_x)^2+\alpha^2t_y^2 \\
        &\Rightarrow  0 = (\alpha^2-1)t_y^2+2\sqrt{\alpha^2-1}(s_x+t_x)t_y+\alpha^2(s_x-t_x)^2-(\alpha^2-1)(s_x+t_x)^2\\
        &\hspace{0.7cm}= (\alpha^2-1)t_y^2+2\sqrt{\alpha^2-1}(s_x+t_x)t_y+\alpha^2s_x^2+\alpha^2t_x^2-2\alpha^2s_xt_x-\alpha^2s_x^2-\alpha^2t_x^2\\
        &\hspace{1.2cm}-2\alpha^2s_xt_x+(s_x+t_x)^2 \\
        &\hspace{0.7cm}= (\alpha^2-1)t_y^2+2\sqrt{\alpha^2-1}(s_x+t_x)t_y-4\alpha^2s_xt_x+(s_x+t_x)^2\\
        &\Rightarrow t_y=\frac{-2\sqrt{\alpha^2-1}(s_x+t_x)\pm\sqrt{4(\alpha^2-1)(s_x+t_x)^2-4(\alpha^2-1)[-4\alpha^2s_xt_x+(s_x+t_x)^2]}}{2(\alpha^2-1)}\\
        &\hspace{0.7cm}= \frac{-\sqrt{\alpha^2-1}(s_x+t_x)\pm\sqrt{(\alpha^2-1)(s_x+t_x)^2+(\alpha^2-1)4\alpha^2s_xt_x-(\alpha^2-1)(s_x+t_x)^2}}{\alpha^2-1}\\
        &\hspace{0.7cm}= \frac{-\sqrt{\alpha^2-1}(s_x+t_x)\pm\sqrt{(\alpha^2-1)4\alpha^2s_xt_x}}{\alpha^2-1} = \frac{-(s_x+t_x)\pm2\alpha\sqrt{s_xt_x}}{\sqrt{\alpha^2-1}}.
    \end{align*}
    
    Let $ (0, b_1) $ and $ (0, b_2) $ be, respectively, the points where $ \pi_{13}(s, t) $ leaves and enters for the second time the region~$R$. We know that $ b_1 > b_2 $. Thus, using Lemma~\ref{lem:path12}, we know that $ \pi_{13}(s,t) $ exists if $ t_y < -\frac{s_x+t_x}{\sqrt{\alpha^2-1}} $.
    Hence, the bisector is given by the curve $ t_y = \frac{-(s_x+t_x)-2\alpha\sqrt{s_xt_x}}{\sqrt{\alpha^2-1}} $.
\end{proof}

\subsection{The source point $s$ lies inside $R$}
\label{sub:spm_The_source_point_lies_inside}

In this case we have shortest paths of type $ \pi_i(s,t) $, for $ i \in \{6,7,8,9,10,12,13\} $. Hence, the equations of the bisectors of the \SPM are given by the sides of $ R $, and bisectors $ b_{6,9}, b_{6,10} $ and $ b_{9,10} $ if $ \alpha < 1 $, and bisectors $ b_{6,7}, b_{7,8}, b_{6,10}, b_{7,10}, b_{8,10} $ and $ b_{12,13} $ if $ 1 < \alpha < \sqrt{2} $. See Lemma~\ref{lemma:bisectorsvisible} and Conjecture~\ref{conj:pi9}.

\section{Conclusion}
\label{sec:conclusions}

We analyzed the WRP when there is only one weighted rectangle $R$, and showed how to obtain the exact shortest path $\pi(s,t)$ and its length when $s$ lies in or on $R$. When both $s$ and $t$ lie outside $R$ the exact solution is unsolvable in the \ACMQ. We obtain similar results in the case where $R$ is a single quadrant. For future work, it would be interesting to find an exact formula within the \ACMQ for the bisectors involving $ \pi_{10}(s,t) $. In addition, we may want to analyze if or how we can generalize these results to other convex shapes.

\acknowledgements
\label{sec:ack}
Work by G. E. and R. I. S. has been supported by project PID2023-150725NB-I00 funded by MICIU/AEI/10.13039/501100011033. G. E. was also funded by a FPU of the Universidad de Alcal\'a.

\nocite{*}
\bibliographystyle{abbrvnat}
\bibliography{wrp_bibliography}
\label{sec:biblio}

\end{document}